\documentclass[10pt,fleqn]{article}
\usepackage{amssymb}
\usepackage{amsfonts}
\usepackage{amsmath}
\usepackage{amsthm}
\usepackage{graphicx}
\usepackage{epsfig}
\usepackage{psfrag}
\usepackage{color}
\usepackage{dsfont}
\makeatletter
\xdef\@endgadget#1{{\unskip\nobreak\hfil\penalty50\hskip1em\hbox{}\nobreak
    \hfil#1\parfillskip=0pt\finalhyphendemerits=0\par}}
\def\@qedsymbol{${}_\blacksquare$}
\def\qed{\@endgadget{\@qedsymbol}}
\newtheorem{lemma}{Lemma}[section]
\newtheorem{theorem}[lemma]{Theorem}

\newtheorem{example}[lemma]{Example}
\newtheorem{definition}[lemma]{Definition}

\newtheorem{proposition}[lemma]{Proposition}
\newtheorem{remark}[lemma]{Remark}
\newcommand{\mR}{\mathbb{R}}
\newcommand{\mL}{\mathbb{L}}

\newcommand{\mP}{\mathbb{P}}
\newcommand{\mZ}{\mathbb{Z}}

\newcommand{\T}{\mathcal{T}}

\newcommand{\Z}{\mathcal{Z}}

\newcommand{\cL}{\mathcal{L}}
\newcommand{\Q}{\mathcal{Q}}

\newcommand{\bq}{\begin{equation}}
\newcommand{\eq}{\end{equation}}

\def\BibTeX{{\rm B\kern-.05em{\sc i\kern-.025em b}\kern-.08em
    T\kern-.1667em\lower.7ex\hbox{E}\kern-.125emX}}

\title{\LARGE \bf Liouville geometry of classical thermodynamics}

\author{Arjan van der Schaft
\thanks{A.J. van der Schaft is with the Bernoulli Institute for Mathematics, Computer
Science and AI, and the Jan C. Willems Center for Systems and Control, University of Groningen, PO Box 407, 9700 AK, the
Netherlands,
        {\tt\small a.j.van.der.schaft@rug.nl}}
}

\begin{document}

\maketitle
\thispagestyle{empty}
\pagestyle{empty}


\section{Introduction}
Starting from Gibbs' fundamental thermodynamic relation, {\it contact geometry} has been recognized as a natural framework for the geometric formulation of classical thermodynamics since the early 1970s \cite{hermann}. This spurred  a series of papers; see e.g. \cite{mrugala1, Mrugala85, mrugala2, Mrugala93, Mrugala00, mrugala3, arnold-gibbs, haslach, eberard, favache, grmela, bravetti15, maschke16, bravetti17, gromov2, ramirez, hudon, simoes, deleon, farantos}, and \cite{bravetti} for a recent introduction and survey. Other geometric work emphasizing the variational formulation of thermodynamics includes \cite{merker,gaybalmaz}.

On the other hand, as discussed in \cite{balian}, the contact-geometric formulation of thermodynamics makes a distinction between the energy and the entropy representation of the same thermodynamic system. By itself this need not be considered as a major flaw since the two representations are conformally equivalent. Nevertheless, it was shown in \cite{balian}, and later in \cite{valparaiso1,valparaiso2,entropy}, that an attractive point of view that is {\it merging} the energy and entropy representation is offered by the extension of contact manifolds to {\it symplectic} manifolds. Compared with the odd-dimensional contact manifold this even-dimensional symplectic manifold has one more degree of freedom, called a gauge variable in \cite{balian}. From a thermodynamics perspective it amounts to replacing the intensive variables by their {\it homogeneous coordinates}. In fact, this {\it symplectization} of contact manifolds is rather well-known in differential geometry \cite{arnold, libermann}; dating back to \cite{herglotz}. As argued in \cite{entropy}, the extension of contact manifolds to symplectic manifolds, in fact to cotangent bundles without zero section, has additional advantages for the geometric formulation of thermodynamics as well. First, it yields a clear distinction between the extensive and intensive variables of the thermodynamic system. Secondly, it enables the definition of {\it port-thermodynamic systems}, which are thermodynamic systems that interact with their environment via either {\it power} or {\it entropy flow} ports. Finally, symplectization has {\it computational} benefits; as was already argued within differential geometry by Arnold \cite{arnold,arnold-contact}.

The present paper aims at providing an in-depth treatment of the resulting geometry of thermodynamic systems, continuing the earlier investigations in \cite{valparaiso1, entropy} and building upon \cite{arnold, arnold-contact, libermann}. Starting point are cotangent bundles without zero section, endowed with their natural one-form; also called the {\it Liouville form}\footnote{Sometimes also called the {\it Poincar\'e-Liouville form}, or tautological form.}. Instead of considering the symplectic geometry derived from the symplectic form $\omega= d \alpha$, where $\alpha$ is the Liouville form, a smaller set of geometric objects will be defined solely based on this Liouville form. The resulting geometry is called {\it Liouville geometry}. In particular, it will be shown how a particular class of Lagrangian submanifolds (called Liouville submanifolds) can be defined as maximal submanifolds on which the Liouville form is zero. Furthermore, a particular type of Hamiltonian vector fields is defined consisting of vector fields which leave the Liouville form invariant. All these geometric objects have the property that they are {\it homogeneous} in the cotangent variables. As a result they are in one-to-one correspondence with objects on the underlying contact manifold (of dimension one less). We will study in detail the generating functions of Liouville submanifolds and the homogeneous Hamiltonian functions of this special type of Hamiltonian vector fields, and relate them to their contact geometry counterparts. Continuing upon \cite{entropy} it will be shown how this leads to the definition of a port-thermodynamic system, and its projection to the contact manifold. Finally we will focus on an {\it additional} homogeneity structure, present in some thermodynamic systems, corresponding to homogeneity in the {\it extensive} variables. This leads to a new geometric view on the classical Gibbs-Duhem relation, and a subsequent projection to an even-dimensional space.

The rest of the paper is structured as follows. In Section 2 it is discussed, using the example of a simple gas, how thermodynamics leads to the study of cotangent bundles over the base space of extensive variables, with cotangent variables being the homogeneous coordinates for the intensive variables. The resulting Liouville geometry of a general cotangent bundle without zero section, and its projection to contact geometry, is studied in Section 3. Then Section 4 provides the definition of port-thermodynamic systems using Liouville geometry, and its projection to a contact-geometric description. Section 5 discusses homogeneity with respect to the extensive variables, the Gibbs-Duhem relation, and its geometric formalization. Finally, Section 6 contains the conclusions.

\section{From thermodynamics to contact and Liouville geometry}
\label{sec:thermo}

In this section we will motivate how classical thermodynamics, starting from Gibbs' thermodynamic relation, naturally leads to contact geometry, and how by considering homogeneous coordinates for the intensive variables this results in Liouville geometry.

\subsection{From Gibbs' fundamental thermodynamic relation to contact geometry}
Consider a simple thermodynamic system such as a mono-phase, single constituent, gas in a confined compartment with volume $V$ and pressure $P$ at temperature $T$. It is well-known that the {\it state properties} of the gas are described by a $2$-dimensional submanifold of the ambient space $\mR^5$ (the {\it thermodynamic phase space}) with coordinates $E$ (energy), $S$ (entropy), $V$, $P$, and $T$. Such a submanifold characterizes the properties of the gas (e.g., an ideal gas, or a Van der Waals gas), and all of them share the following property. Define the Gibbs one-form on the thermodynamic phase space $\mR^5$ as
\bq
\label{Gibbs}
\theta:=dE - TdS +PdV
\eq
Then $\theta$ is {\it zero} restricted to the submanifold characterizing the state properties. This is called {\it Gibbs' fundamental thermodynamic relation}. It implies that the {\it extensive} variables $E,S,V$ and the {\it intensive} variables $T,P$ are related in a specific way. Geometrically this is formalized by noting that the Gibbs one-form $\theta$ defines a {\it contact form} on $\mR^5$, and that any submanifold $L$ capturing the state properties of the thermodynamic system is a submanifold of maximal dimension restricted to which the contact form $\theta$ is zero. Such submanifolds are called {\it Legendre submanifolds} of the {\it contact manifold} $(\mR^5, \theta)$.

By expressing the extensive variable $E$ as a function $E=E(S,V)$ of the two remaining extensive variables $S$ and $V$, Gibbs' fundamental relation implies that 
the Legendre submanifold $L$ specifying the state properties is given as
\bq
\label{L1}
L=\{(E,S,V,T,P) \mid E=E(S,V), T= \frac{\partial E}{\partial S}, -P= \frac{\partial E}{\partial V} \}
\eq
Hence $L$ is completely described by the energy function $E(S,V)$, whence the name {\it energy representation} for \eqref{L1}. On the other hand, there are other ways to represent $L$. If $L$ is parametrizable by the variables $T,V$ (instead of $S,V$ as in \eqref{L1}), then one defines the {\it partial Legendre transform} of $E(S,V)$ with respect to $S$ as
\bq
A(T,V) :=  E(V,S) -TS, \quad T= \frac{\partial E}{\partial S}(S,V),
\eq
where $S$ is solved from $T= \frac{\partial E}{\partial S}(S,V)$. Then $L$ is also described as
\bq
\label{L2}
L=\{(E,S,V,T,P) \mid E=A(T,V) - T \frac{\partial A}{\partial T}, S= - \frac{\partial A}{\partial T}, -P= \frac{\partial A}{\partial V} \}
\eq
$A$ is known as the {\it Helmholtz free energy}, and is one of the thermodynamic potentials derivable from the energy function $E(S,V)$; see e.g. \cite{fermi}. Two other possible parametrizations of $L$ (namely by $S,P$, respectively by $T,P$) correspond to two more thermodynamic potentials, namely the enthalpy $H(S,P)$ and the Gibbs' free energy $G(T,P)$, resulting in similar expressions for $L$. 

In general \cite{arnold, libermann}, a {\it contact manifold} $(M, \theta)$ is an odd-dimensional manifold equipped with a {\it contact form} $\theta$. A one-form $\theta$ on a $(2n+1)$-dimensional manifold $M$ is a contact form if and only around any point in $M$ we can find coordinates $(q_0,q_1, \cdots,q_n, \gamma_1, \cdots, \gamma_n)$ for $M$, called {\it Darboux coordinates}, such that
\bq
\label{darboux}
\theta = dq_0 - \sum_{j=1}^n \gamma_j dq_j
\eq
Equivalently, $\theta$ is a contact form if $\theta \wedge (d \theta)^n$ is nowhere zero on $M$.
A {\it Legendre submanifold} of a contact manifold $(M,\theta)$ is a submanifold of maximal dimension restricted to which the contact form $\theta$ is zero. The dimension of any Legendre submanifold of a $(2n+1)$-dimensional contact manifold is equal to $n$. 
%

In fact, we will use throughout this paper the slightly {\it generalized definition} of a contact manifold as given in e.g. \cite{arnold}, where the contact form $\theta$ {\it is only required to be defined locally}. What counts is the {\it contact distribution}; the $2n$-dimensional subspace of the tangent space at any point of $M$ defined by the {\it kernel} of the contact form $\theta$ at this point. This turns out to be the appropriate concept for the thermodynamic phase space being a contact manifold\footnote{Contact manifolds for which the contact form $\theta$ is defined globally are sometimes called {\it exact contact manifolds}.}.

\medskip

Apart from the above parametrizations of the Legendre submanifold $L$, corresponding to an energy function $E(S,V)$ and its Legendre transforms, there is still {\it another}, although very similar, way of describing $L$. 
This alternative option is motivated from a modeling point of view. Namely, often thermodynamic systems are formulated by first listing the {\it balance laws} for the extensive variables apart from the entropy $S$, and then expressing $S$ as a function $S=S(E,V)$. This leads to the {\it entropy representation} of the submanifold $L \subset \mR^5$, given as
\bq
\label{L3}
L:= \{(E,S,V,T,P) \mid S=S(E,V), \frac{1}{T}= \frac{\partial S}{\partial E}, \frac{P}{T}= \frac{\partial S}{\partial V} \}
\eq
Analogously the case of the energy representation $E=E(S,V)$, one may consider thermodynamic potentials obtained by partial Legendre transform of $S(E,V)$. Geometrically the entropy representation corresponds to the {\it modified} Gibbs contact form
\bq
\label{Gibbs1}
\widetilde{\theta}:=dS -\frac{1}{T}dE - \frac{P}{T}dV,
\eq
which is obtained from the original Gibbs contact form $\theta$ in \eqref{Gibbs} by division by $-T$ (called {\it conformal equivalence}). In this way the Gibbs fundamental relation is rewritten as $\widetilde{\theta}|_L=0$, and the intensive variables become $\frac{1}{T},\frac{P}{T}$.

\subsection{From contact to Liouville geometry}
The contact-geometric view on thermodynamics, directly motivated by Gibbs' fundamental thermodynamic relation, has two shortcomings:\\
(1) Switching from the energy representation $E=E(S,V)$ to the entropy representation $S=S(E,V)$ corresponds to replacing the Gibbs form $\theta$ by the modified Gibbs form $\widetilde{\theta}$ in \eqref{Gibbs1}, and thus leads to a similar, but {\it different}, contact-geometric description.\\
(2) The contact-geometric description does not make a clear distinction between, on the one hand, the extensive variables $E,S,V$ and, on the other hand, the intensive variables $T,-P$ (energy representation), or $\frac{1}{T},\frac{P}{T}$ (entropy representation). In fact, given a contact form $\theta$ there are many Darboux coordinates $q_0,q_1,q_2,p_1,p_2$ for $\mR^5$ such that $\theta=dq_0 - p_1dq_1 - p_2dq_2$, where $q_0,q_1,q_2$ are {\it not} necessarily obtained by a transformation of only the extensive variables $E,S,V$.\\

The way to remedy these shortcomings is to {\it extend} the contact manifold by one extra dimension to a symplectic manifold, in fact a {\it cotangent bundle}, with an additional {\it homogeneity} structure. This construction is rather well-known in differential geometry \cite{arnold, libermann}, but was advocated within a thermodynamics context only in \cite{balian}, and followed up in \cite{valparaiso1, entropy}. For a simple thermodynamic system with extensive variables $E,S,V$ and intensive variables $T,-P$, the construction amounts to replacing the intensive variables $T,-P$ by their {\it homogeneous coordinates} $p_E,p_S,p_V$ with $p_E \neq 0$, i.e., 
\bq
T= \frac{p_S}{-p_E}, \; -P= \frac{p_V}{-p_E}
\eq
Equivalently, the intensive variables $\frac{1}{T}, \frac{P}{T}$ in the {\it entropy representation} are represented as
\bq
\frac{1}{T} = \frac{p_E}{-p_S}, \; \frac{P}{T}= \frac{p_V}{-p_S}
\eq
This means that the {\it two} contact forms $\theta=dE - TdS +PdV$ and $\widetilde{\theta}=dS -\frac{1}{T}dE - \frac{P}{T}dV$ are replaced by a {\it single} symmetric expression, namely by
\bq
\alpha:=p_EdE + p_SdS + p_VdV,
\eq
The one-form $\alpha$ is nothing else than the canonical {\it Liouville one-form} on the cotangent bundle $T^*\mR^3$, with $\mR^3$ the space of extensive variables $E,S,V$. Thus the thermodynamic phase space $\mR^5$ has been replaced by $T^*\mR^3$. More precisely, by definition of homogeneous coordinates the vector $(p_E,p_S,p_V)$ is different from the zero vector, and hence the space with coordinates $E,S,V,p_E,p_S,p_V$ is actually the cotangent bundle $T^*\mR^3$ {\it minus} its zero section; denoted as $\T^*\mR^3$.

Any $2$-dimensional Legendre submanifold $L \subset \mR^5$ describing the state properties is now replaced by a $3$-dimensional submanifold $\cL \subset \T^*\mR^3$, given as
\bq
\cL=\{(E,S,V,p_E,p_S,p_V) \in \T^*\mR^3 \mid (E,S,V, \frac{p_S}{-p_E}, \frac{p_V}{-p_E}) \in L  \}
\eq
It turns out that $\cL$ is a {\it Lagrangian submanifold} of $\T^*\mR^3$ with symplectic form $\omega:=d\alpha$, with an additional property of {\it homogeneity}. Namely, whenever $(E,S,V,p_E,p_S,p_V) \in \cL$, then also $(E,S,V,\lambda p_E,\lambda p_S, \lambda p_V) \in \cL$, for any non-zero $\lambda \in \mR$. Such Lagrangian submanifolds turn out to be fully characterized as maximal manifolds restricted to which the Liouville one-form $\alpha=p_EdE + p_SdS + p_VdV$ is zero, and will thus be called {\it Liouville submanifolds} of $\T^*\mR^3$.
As we will see in the next section the extension of contact manifolds to cotangent bundles, replacing the intensive variables by their homogeneous coordinates, also leads to a natural homogeneous Hamiltonian dynamics on the extended space $\T^*\mR^3$. This does not only facilitate the analysis, but has clear computational advantages as well. In fact, all computations become standard operations on cotangent bundles and in Hamiltonian dynamics. In the words of Arnold \cite{arnold-contact}: one is advised to calculate symplectically (but to think rather in terms of contact geometry). 

\medskip

All of this is immediately extended from the thermodynamic phase space $\mR^5$ with coordinates $E,S,V,T,P$ to {\it general} thermodynamic phase spaces. For instance, in the case of multiple chemical species the Gibbs form $\theta$ extends to $dE - TdS +PdV - \sum_k \mu_kdN_k$, where $N_k$ and $\mu_k$, $k=1, \cdots,s,$ are the mole numbers, respectively, chemical potentials of the $k$-th species. Correspondingly, the thermodynamic phase $\mR^5 \times \mR^{2s}$ is replaced by the cotangent bundle without zero-section $\T^*\mR^{3 + s}$, with extensive variables $E,S,V,N_1, \cdots, N_s$ and Liouville form
\bq
p_EdE + p_SdS + p_VdV + p_1dN_1 + \cdots + p_sdN_s,
\eq
where $\mu_1 =\frac{p_1}{-p_E}, \cdots, \mu_s =\frac{p_s}{-p_E}$.

\section{Liouville geometry}

This section is concerned with the general definition and analysis of geometric objects on the cotangent bundle without zero section, which project to the underlying contact manifold. Since everything is based on the Liouville form this will be called {\it Liouville geometry}. In particular, we will deal with Liouville submanifolds and homogeneous Hamiltonian dynamics.

\subsection{Cotangent bundles and the canonical contact manifold}
In the previous section it was indicated how the thermodynamic phase space can be extended to a cotangent bundle, without its zero section, by the use of homogeneous coordinates for the intensive variables. Furthermore, it was shown how in this way the energy and entropy representation are unified, and how this provides a geometric definition of extensive and intensive variables. Conversely, in this subsection we will {\it start} with a general cotangent bundle without zero section, and show how this leads to the {\it canonical contact manifold} serving as thermodynamic phase space.

Consider a thermodynamic system with total space of extensive variables, including energy $E$ and entropy $S$, given by the manifold $\Q$. Then consider the cotangent bundle $\T^*\Q$ without its zero section. The {\it Liouville one-form} $\alpha$ on $\T^*\Q$ is defined as follows.
Consider $\eta \in \T^*\Q, X \in T_{\eta}\T^*\Q$, and define
\bq
\alpha_{\eta}(X) := \eta (\mathrm{pr}_*X),
\eq
where $\mathrm{pr}:\T^*\Q \to \Q$ is the bundle projection. Then $\omega:=d\alpha$, with $d$ exterior derivative, is the canonical {\it symplectic form} on $\T^*\Q$. Furthermore, the {\it Euler vector field} $Z$ is defined as the unique vector field satisfying
\bq
\label{Z}
d\alpha(Z, \cdot) = \alpha
\eq
This also implies $\mL_Z \alpha = \alpha$, with $\mL$ denoting Lie derivative.

In coordinates $\alpha, \omega$ and $Z$ take the following simple form. Let $\dim \Q=n+1$, with local coordinates $q_0,\cdots,q_n$, and let $p_0,\cdots,p_n$ be the corresponding coordinates for the cotangent spaces $T_q^*\Q$. Then 
\bq
\alpha = \sum_{i=0}^{n}p_idq_i, \quad \omega= \sum_{i=0}^{n}dp_i \wedge dq_i, \quad Z= \sum_{i=0}^{n}p_i \frac{\partial}{\partial p_i}
\eq
Based on $\T^*\Q$ we may define a {\it canonical contact manifold} in the following way \cite{arnold}. For each $q\in \Q$ and each cotangent space $T^*_q\Q$ consider the {\it projective space} $\mP(T^*_q\Q)$, given as the set of rays in $T^*_q\Q$, that is, all the non-zero multiples of a non-zero cotangent vector. Thus the projective space $\mP(T^*_q\Q)$ has dimension $n$, and there is a canonical projection $\pi_q: \T^*_q\Q \to \mP(T^*_q\Q)$, where $\T^*_q\Q$ denotes the cotangent space without its zero vector.  
The fiber bundle of the projective spaces $\mP(T^*_q\Q)$, $q \in \Q$, over the base manifold $\Q$ will be denoted by $\mP(T^*\Q)$. Furthermore, denote the bundle projection obtained by considering $\pi_q: \T^*_q\Q \to \mP(T^*_q\Q)$ for every $q\in\Q$ by $\pi: \T^*\Q \to \mP(T^*\Q)$. 

As detailed in \cite{arnold,arnold-contact,entropy,valparaiso1}, $\mP(T^*\Q)$ defines a {\it canonical}\footnote{In the sense that any other $(2n+1)$-dimensional contact manifold is locally {\it contactomorphic} to $\mP(T^*\Q)$ \cite{arnold,libermann}.} contact manifold of dimension $2n+1$. The contact manifold $\mP(T^*\Q)$ will serve as the {\it thermodynamic phase space} for the thermodynamic system with space of external variables $\Q$. 

Given natural coordinates $q_0,\cdots,q_n, p_0,\cdots,p_n$ for $\T^*\Q$, we may select different sets of local coordinates for $\mP(T^*\Q)$ and corresponding different expressions of the projection $\pi: \T^*_q\Q \to \mP(T^*_q\Q)$. In fact, whenever $p_0 \neq 0$ we may express the projection $\pi_q: \T^*_q\Q \to \mP(T^*_q\Q)$ by the map 
\bq
(p_0,p_1, \cdots, p_n) \mapsto (\gamma_1,\cdots,\gamma_n)
\eq
where
\bq
\gamma_1=\frac{p_1}{-p_0}, \cdots, \gamma_n=\frac{p_n}{-p_0}
\eq
This means that
\bq
\alpha = p_0dq_0 + p_1dq_1 + \cdots + p_ndq_n=p_0\big(dq_0 - \gamma_1dq_1 \cdots - \gamma_ndq_n \big) =: p_0\theta,
\eq
with $\theta$ a locally defined contact form on $\mP(T^*\Q)$. Clearly, the same can be done for any of the other coordinates $p_i$, defining different contact forms. For example, if $p_1 \neq 0$ we may express $\pi_q: \T^*_q\Q \to \mP(T^*_q\Q)$ also by the map 
\bq
(p_0,p_1, \cdots, p_n) \mapsto (\widetilde{\gamma}_0, \widetilde{\gamma}_2, \cdots,\widetilde{\gamma}_n),
\eq
where
\bq
\widetilde{\gamma}_0=\frac{p_0}{-p_1}, \widetilde{\gamma}_2=\frac{p_2}{-p_1}, \cdots, \widetilde{\gamma}_n=\frac{p_n}{-p_1},
\eq
so that 
\bq
\alpha = p_1\big(dq_1 - \widetilde{\gamma}_0dq_0 - \widetilde{\gamma}_2dq_2 \cdots - \widetilde{\gamma}_ndq_n \big) =: p_1\widetilde{\theta}
\eq
In the thermodynamics context of Section \ref{sec:thermo}, with $q_0=E,q_1=S$, and thus $p_0=p_E,p_1=p_S$, the first option corresponds to the {\it energy representation} and the second to the {\it entropy representation}.


Importantly, there is a direct {\it correspondence} between all geometric objects (functions, Legendre submanifolds, vector fields) on the contact manifold $\mP(T^*\Q)$ with the same objects on $\T^*\Q$ endowed with an additional homogeneity property in the $p$ variables. A key element in this is Euler's theorem on homogeneous functions; see e.g. \cite{entropy}. 
\begin{definition}
Let $r \in \mZ$. A function $K: \T^*\Q \to \mR$ is called homogeneous of degree $r$ in $p$ if
\bq
K(q, \lambda p) = \lambda^r K(q, p), \quad \mbox{ for all } \lambda \neq 0
\eq
\end{definition}
\begin{theorem}[Euler's homogeneous function theorem]\label{Euler}
A differentiable function $K: \T^*\Q \to \mR$ is homogeneous of degree $r$ in $p$ if and only if
\bq
\sum_{i=0}^{n} p_i\frac{\partial K}{\partial p_i}(q,p)= rK(q,p), \quad \mbox{ for all } (q,p) \in \T^*\Q
\eq
Moreover, if $K$ is homogeneous of degree $r$ in $p$, then all its derivatives \\$\frac{\partial K}{\partial p_i}(q,p), i=0,1,\cdots,n,$ are homogeneous of degree $r-1$ in $p$.

Furthemore $K: \T^*\Q \to \mR$ is homogeneous of degree $0$ in $p$ if and only if $\mL_ZK=0$, and homogeneous of degree $1$ in $p$ if and only if $\mL_ZK=K$, where $Z$ is the Euler vector field and $\mL$ denotes Lie derivation.
\end{theorem}
Since until Section \ref{sec:homex} homogeneity will always refer to homogeneity in the $p$-variables we will often simply talk about 'homogeneity'.

Obviously, functions $K: \T^*\Q \to \mR$ which are homogeneous of degree $0$ in $p$ are those functions which project under $\pi$ to functions on $\mP(T^*\Q)$, i.e., $K=\pi^*\widehat{K}$ with $\widehat{K}: \mP(T^*\Q) \to \mR$. In the next two subsections we will consider two more classes of objects which project to $\mP(T^*\Q)$.

\subsection{Liouville submanifolds}

Legendre submanifolds of the canonical thermodynamic phase space $\mP(T^*\Q)$ are in one-to-one correspondence with Liouville submanifolds\footnote{Previously called {\it homogeneous Lagrangian submanifolds} in \cite{entropy}.} of $\T^*\Q$, defined as follows.
\begin{definition} 
A submanifold $\cL \subset \T^*\Q$ is called a {\it Liouville submanifold} if the Liouville form $\alpha$ restricted to $\cL$ is zero and $\dim \cL= \dim \Q$.
\end{definition}
Recall that $\cL$ is a Lagrangian submanifold of $\T^*\Q$ if $\omega=d\alpha$ is zero on $\cL$ and $\dim \cL= \dim \Q$ (or, equivalently, $\omega$ is zero on $\cL$ and $\cL$ is maximal with respect to this property.) The following proposition shows that Liouville submanifolds are actually Lagrangian submanifolds of $\T^*\Q$ with an additional homogeneity property.  
\begin{proposition}
\label{prop:liouville}
$\cL \subset \T^*\Q$ is a Liouville submanifold if and only if $\cL$ is a Lagrangian submanifold of the symplectic manifold $(\T^*\Q, \omega)$ with the property that
\bq
\label{Liouville}
(q,p) \in \cL \Rightarrow (q,\lambda p) \in \cL
\eq
for every $0\neq \lambda \in \mR$.
\end{proposition}
\begin{proof} First of all note that the homogeneity property \eqref{Liouville} is equivalent to {\it tangency} of the Euler vector field $Z$ to $\cL$.\\
(Only if) By Palais' formula (see e.g. \cite{abraham}, Proposition 2.4.15)
\bq
d \alpha (X_1,X_2) = \mL_{X_1}(\alpha (X_2)) -  \mL_{X_2}(\alpha (X_1)) - \alpha \left( [X_1,X_2] \right)
\eq
for any two vector fields $X_1,X_2$. Hence, for any $X_1,X_2$ tangent to $\cL$ we obtain $d \alpha (X_1,X_2)=0$, implying that $\cL$ is a Lagrangian submanifold. Furthermore, by \eqref{Z}
\bq
\label{alpha}
d\alpha(Z,X)=\alpha (X)=0,
\eq
for all vector fields $X$ tangent to $\cL$. Because $\cL$ is a Lagrangian submanifold this implies that $Z$ is tangent to $\cL$ (since a Lagrangian submanifold is a {\it maximal} submanifold restricted to which $\omega=d \alpha$ is zero.) \\
(If). If $\cL$ is Lagrangian and satisfies \eqref{Liouville}, then $Z$ is tangent to $\cL$, and thus \eqref{alpha} holds for all vector fields $X$ tangent to $\cL$, implying that $\alpha$ is zero restricted to $\cL$.
\end{proof}
\begin{remark}
It also follows that $\cL \subset \T^*\Q$ is a Liouville submanifold if and only if it is a {\it maximal} submanifold on which $\alpha$ is zero.
\end{remark}
Liouville submanifolds of $\T^*\Q$ are in one-to-one correspondence with {\it Legendre submanifolds} of the canonical contact manifold $\mP(T^*\Q)$.
Recall that a submanifold of a $(2n+1)$-dimensional contact manifold is a Legendre submanifold \cite{arnold, libermann} if the locally defined contact form $\theta$ is zero restricted to it, and its dimension is equal to $n$ (the maximal dimension of a submanifold on which $\theta$ is zero). 

\begin{proposition}[\cite{libermann}, Proposition 10.16, \cite{entropy}]
Consider the projection $\pi: \T^*\Q \to \mP (T^*\Q)$. Then $\widehat{\cL} \subset \mP (T^*\Q)$ is a Legendre submanifold if and only if $\cL:=\pi^{-1}(\widehat{\cL}) \subset \T^*\Q$ is a Liouville submanifold. Conversely, any Liouville submanifold $\cL \subset \T^*\Q$ is of the form $\pi^{-1}(\widehat{\cL})$ for some Legendre submanifold $\widehat{\cL}$.
\end{proposition}

This implies as well a one-to-one correspondence between {\it generating functions} of Legendre submanifolds $\widehat{\cL} \subset \mP (T^*\Q)$ and generating functions of Liouville submanifolds $\cL \subset \T^*\Q$ with $\pi^{-1}(\widehat{\cL})$.
Recall from \cite{libermann, arnold} that any Legendre submanifold $\widehat{\cL} \subset \mP (T^*\Q)$ with Darboux coordinates $q_0,q_1,\cdots,q_n,\gamma_1,\cdots, \gamma_n$ can be represented as
\bq
\label{20}
\widehat{\cL} = \{(q_0,q_1,\cdots,q_n, \gamma_1, \cdots, \gamma_n) \mid q_0=\widehat{F} - \gamma_J\frac{\partial \widehat{F}}{\partial \gamma_J}, 
\, q_J = - \frac{\partial \widehat{F}}{\partial \gamma_J}, \, \gamma_I = \frac{\partial \widehat{F}}{\partial q_I}\}
\eq
for some disjoint partitioning $I \cup J=\{1, \cdots, n \}$ and some function $\widehat{F}(q_I,\gamma_J)$, called a {\it generating function} for $\widehat{\cL}$. Here $\gamma_J$ is the vector with elements $\gamma_\ell=\frac{p_\ell}{-p_0}, \ell \in J$, and $\gamma_J \frac{\partial \widehat{F}}{\partial \gamma_J}$ is shorthand notation for $\sum_{\ell \in J}\gamma_\ell \frac{\partial \widehat{F}}{\partial \gamma_\ell}$. Conversely any submanifold $\widehat{\cL}$ as given in \eqref{20}, for any partitioning $I \cup J=\{1, \cdots, n \}$ and function $\widehat{F}(q_I,\gamma_J)$, is a Legendre submanifold. This implies that the corresponding Liouville submanifold $\cL=\pi^{-1}(\widehat{\cL})$ is given as
\bq
\label{21}
\cL= \{(q_0,\cdots,q_n,p_0,\cdots,p_n) \mid 
q_0=- \frac{\partial F}{\partial p_0}, \, q_J = - \frac{\partial F}{\partial p_J}, \, p_I = \frac{\partial F}{\partial q_I}\},
\eq
where 
\bq
\label{22}
F(q_I,p_0,p_J) :=-p_0\widehat{F}(q_I, \frac{p_J}{-p_0})
\eq
This is immediately verified by exploiting the identities
\bq
\label{identities}
\begin{array}{l}
- \frac{\partial F}{\partial p_0} = \widehat{F}(q_I, -\frac{p_J}{p_0}) + p_0 \frac{\partial \widehat{F}}{\partial \gamma_J}(q_I, -\frac{p_J}{p_0})\frac{p_J}{p_0^2} = \widehat{F}(q_I, \gamma_J) - \gamma_J \frac{\partial \widehat{F}}{\partial \gamma_J} \\[2mm]
\frac{\partial F}{\partial p_J}= -p_0\frac{\partial \widehat{F}}{\partial \gamma_J}\cdot\frac{1}{-p_0}=\frac{\partial \widehat{F}}{\partial \gamma_J}, \quad 
\frac{\partial F}{\partial q_I} = -p_0\frac{\partial \widehat{F}}{\partial q_I}= -p_0\gamma_I=p_I
\end{array}
\eq

Thus $F(q_I,p_0,p_J)$ is a generating function of $\cL$. 
Conversely, any Liouville submanifold as in \eqref{21} for some $p_0$ (possibly after renumbering the index set $\{0,1,\cdots,n\}$) and generating function $F$ as given in \eqref{22} for some $\widehat{F}(q_I,\gamma_J)$, with $I \cup J=\{1, \cdots, n \}$ and $\gamma_J= -\frac{p_J}{p_0}$ defines a Liouville submanifold of $\T^*\Q$.

Note that the generating function $F(q_I,p_0,p_J)=-p_0\widehat{F}(q_I, \frac{p_J}{-p_0})$ as in \eqref{22} for the Liouville submanifold $\cL$ is {\it homogeneous of degree} $1$ in $p$. The correspondence \eqref{22} between the generating function $F(q_I,p_0,p_J)$ of the Liouville submanifold $\cL=\pi^{-1}(\widehat{\cL})$ and the generating function $\widehat{F}(q_I, \gamma_J)$ of the Legendre submanifold $\widehat{\cL}$ is of a well-known type in the theory of homogeneous functions. Indeed, for any function $K(q,p)$ that is homogeneous of degree $1$ in $p$, we can define
\bq
\label{hom1}
\widehat{K}(q,\gamma_1, \cdots,\gamma_n):= K(q,-1, \gamma_1, \cdots,\gamma_n),
\eq
implying that
\bq
\label{hom2}
K(q,p_0,p_1,\cdots,p_n)= -p_0\widehat{K}(q,\frac{p_1}{-p_0}, \cdots, \frac{p_n}{-p_0} )
\eq
Finally note that the correspondence between the Liouville submanifold $\cL$ and the Legendre submanifold $\widehat{\cL}$ and their generating functions can be obtained for {\it any numbering} of the set $\{0,1,\cdots,n\}$, and thus for any choice of $p_0$. This provides other coordinatizations of the {\it same} Legendre submanifold $\widehat{\cL} \subset \mP (T^*\Q)$. The representation of $\widehat{\cL}$ either in energy or in entropy representation is an example of this.

\subsection{Homogeneous Hamiltonian and contact vector fields}
For any function $K: \T^*\Q \to \mR$ the \emph{Hamiltonian vector field} $X_K$ on $\T^*\Q$ is defined by the standard Hamiltonian equations
\bq
\label{ham}
\dot{q}_i = \frac{\partial K}{\partial p_i}(q,p), \quad \dot{p}_i = - \frac{\partial K}{\partial q_i}(q,p), \quad i=0,1\cdots, n,
\eq
or equivalently, $\omega (X_K,-)= - dK$.
Note that since $d\alpha (Z,\cdot)=\alpha$, we have $\alpha(X_K)=d\alpha( Z,X_K) = \mL_{Z}K=K$. Hence a Hamiltonian $K$ is homogeneous of degree $1$ in $p$ if and only if 
\bq
\label{9}
\alpha(X_K)=K
\eq
Furthermore
\begin{proposition}
\label{prop:homham}
If $K: \T^*\Q \to \mR$ is homogeneous of degree $1$ in $p$ then its Hamiltonian vector field $X_K$ satisfies
\bq
\mL_{X_K} \alpha =0
\eq
Conversely, if the vector field $X$ satisfies $\mL_{X}\alpha=0$, then $X=X_K$ where the function $K:=\alpha(X)$ is homogeneous of degree $1$ in $p$.
\end{proposition}
\begin{proof}
By Cartan's formula, with $\mL$ denoting Lie derivative and $i$ contraction, 
\bq
\label{11}
\mL_{X}\alpha = i_{X} d\alpha + d i_{X}\alpha =i_{X} d\alpha + d \left(\alpha (X)\right)
\eq
If $K$ is homogeneous of degree $1$ in $p$ then by \eqref{9} $i_{X_K} d\alpha + d \left(\alpha (X_K) \right)= -dK +dK=0$, implying by \eqref{11} that $\mL_{X_K}\alpha=0$. Conversely, if $\mL_{X}\alpha=0$, then \eqref{11} yields $i_{X} d\alpha + d \left(\alpha (X) \right)$, implying that $X=X_K$ with $K=\alpha(X)$, which by \eqref{9} is homogeneous of degree $1$ in $p$.
\end{proof}
Thus the Hamiltonian vector fields with a Hamiltonian homogeneous of degree $1$ in $p$ are precisely the vector fields that leave the Liouville form $\alpha$ invariant. 
For simplicity of exposition the Hamiltonians $K: \T^*\Q \to \mR$ that are homogeneous of degree $1$ in $p$, and their corresponding Hamiltonian vector fields $X_K$, will be simply called {\it homogeneous} in the sequel. 

Note that by Theorem \ref{Euler} (Euler's theorem) the expressions $\frac{\partial K}{\partial p_i}(q,p), i=0,1\cdots, n,$ are homogeneous of degree $0$ in $p$ since $K$ is homogeneous of degree $1$ in $p$. Hence the dynamics of the extensive variables $q$ in \eqref{ham} is invariant under scaling of the $p$-variables, and thus expressible as a function of $q$ and the intensive variables $\gamma$.
In fact, any homogeneous Hamiltonian vector field projects to a \emph{contact vector field} on the thermodynamic phase space $\mP(T^*\Q)$, and conversely any contact vector field on $\mP(T^*\Q)$ is the projection of a homogeneous Hamiltonian vector field on $\T^*\Q$. This can be seen from the following computations. Consider a homogeneous Hamiltonian vector field $X_K$. Since $K$ is homogeneous of degree $1$ in $p$ we can write as in \eqref{hom2} $K(q,p)= -p_0\widehat{K}(q,\frac{p_1}{-p_0}, \cdots, \frac{p_n}{-p_0} )$, with $\widehat{K}(q,\gamma)$ as defined in \eqref{hom1}. This means that the equations \eqref{ham} of the Hamiltonian vector field $X_K$ take the form
\bq
\begin{array}{rcl}
\dot{q}_0 &= &  - \widehat{K}(q,\gamma) -p_0 \sum_{\ell=1}^n 
\frac{\partial \widehat{K}}{\partial \gamma_\ell}(q,\gamma)\cdot -\frac{p_\ell}{p_0^2}   =  - \widehat{K}(q,\gamma)  + \sum_{\ell=1}^n\gamma_\ell 
\frac{\partial \widehat{K}}{\partial \gamma_\ell}(q,\gamma) \\[3mm]
\dot{q}_j &= & -p_0 \frac{\partial \widehat{K}}{\partial \gamma_j}(q,\gamma)\cdot \frac{1}{-p_0}= \frac{\partial \widehat{K}}{\partial \gamma_j}(q,\gamma), \quad j=1, \cdots,n \\[3mm]
\dot{p}_i &= & p_0 \frac{\partial \widehat{K}}{\partial q_i}(q,\gamma), \qquad \qquad \qquad \qquad \quad \;  i=0, \cdots,n
\end{array}
\eq
where $\gamma_j=\frac{p_j}{-p_0}, j=1, \cdots,n$.
Combining with
\bq
\label{dot}
\dot{\gamma}_j= \frac{1}{-p_0}\dot{p}_j + \frac{p_j}{p^2_0}\dot{p}_0, \quad j=1, \cdots,n,
\eq
this yields the following projected dynamics on the contact manifold $\mP(T^*\Q)$ with coordinates $(q,\gamma)$
\bq
\label{contact}
\begin{array}{rcll}
\dot{q}_0 &= &   \sum_{\ell=1}^n\gamma_\ell
\frac{\partial \widehat{K}}{\partial \gamma_\ell}(q,\gamma) - \widehat{K}(q,\gamma) &\\[3mm]
\dot{q}_j &= &  \frac{\partial \widehat{K}}{\partial \gamma_j}(q,\gamma), & j=1, \cdots,n \\[3mm]
\dot{\gamma}_j &= &   - \frac{\partial \widehat{K}}{\partial q_j}(q,\gamma) - \gamma_j \frac{\partial \widehat{K}}{\partial q_0}(q,\gamma), & j=1 \cdots,n
\end{array}
\eq
This is recognized as the {\it contact vector field} \cite{libermann} with {\it contact Hamiltonian} $\widehat{K}$. Indeed, given a contact form $\theta$ the contact vector field $X_{\widehat{K}}$ with contact Hamiltonian $\widehat{K}$ is defined through the relations\footnote{Here the sign convention of \cite{bravetti} is followed.}
\bq
\label{contact1}
\mL_{X_{\widehat{K}}} \theta = \rho_{\widehat{K}} \theta, \quad -{\widehat{K}}=\theta(X_{\widehat{K}})
\eq
for some function $\rho_{\widehat{K}}$ (depending on $\widehat{K}$). The first equation in \eqref{contact1} expresses the condition that the contact vector field leaves the contact distribution (the kernel of the contact form $\theta$) invariant. Equations \eqref{contact1} for $\theta=dq_0 - \gamma_1dq_1  \cdots - \gamma_ndq_n$ and $\widehat{K}(q,\gamma)$ can be seen to yield the same equations as in \eqref{contact}; see \cite{libermann,eberard} for details. Conversely, any contact vector field with contact Hamiltonian $\widehat{K}(q,\gamma)$ defines a homogeneous Hamiltonian vector field on $\T^*\Q$ with homogeneous Hamiltonian $-p_0\widehat{K}(q,\frac{p_1}{-p_0}, \cdots, \frac{p_n}{-p_0})$. 
As before, the coordinate expression \eqref{contact} of the contact vector field depends on the numbering of the homogeneous coordinates $p_0,p_1,\cdots,p_n$; i.e., the choice of $p_0$. In the thermodynamics context this is again illustrated by the choice of either the energy or entropy representation (corresponding to choosing $p_0=p_E$ or $p_0=p_S$).
 
The projectability of any homogeneous Hamiltonian vector field $X_K$ to a contact vector field $X_{\widehat{K}}$ on $\mP(T^*\Q)$ also follows from the following proposition, and the fact that the projection $\pi: \T^*\Q \to \mP(T^*\Q)$ is along the Euler vector field $Z$.
\begin{proposition}
\label{prop:projection}
Any homogeneous Hamiltonian vector field $X_K$ satisfies \\$[X_K,Z]=0$.
\end{proposition}
\begin{proof}
By \cite{abraham}(Table 2.4-1)
\bq
i_{[X_K,Z]}d \alpha = \mL_{X_K} i_Z d\alpha - i_Z \mL_{X_K} d \alpha = \mL_{X_K} \alpha - i_Z d \mL_{X_K} \alpha = 0 -0=0,
\eq
since $\mL_{X_K} \alpha=0$. Because $\omega=d \alpha$ is a symplectic form this implies $[X_K,Z]=0$.
\end{proof}
Although homogeneous Hamiltonian vector fields are in one-to-one correspondence with contact vector fields, typically {\it computations} for homogeneous Hamiltonian vector fields are much easier than the corresponding computations for their contact vector field counterparts. First note the following properties proved in \cite{entropy,valparaiso1}.
\begin{proposition}\label{app:poisson}
Consider the Poisson bracket $\{K_1,K_2\}$ of functions $K_1,K_2$ on $\T^*\Q$ defined with respect to the symplectic form $\omega=d\alpha$. Then
\begin{enumerate}
\item[(a)]
If $K_1,K_2$ are both homogeneous of degree $1$ in $p$, then also $\{K_1,K_2\}$ is homogeneous of degree $1$ in $p$.
\item[(b)]
If $K_1$ is homogeneous of degree $1$ in $p$, and $K_2$ is homogeneous of degree $0$ in $p$, then $\{K_1,K_2\}$ is homogeneous of degree $0$ in $p$.
\item[(c)]
If $K_1,K_2$ are both homogeneous of degree $0$ in $p$, then $\{K_1,K_2\}$ is zero.
\end{enumerate}
\end{proposition}
Using property $(a)$ we may define the following bracket
\bq
\{\widehat{K}_1,\widehat{K}_2\}_J:=\widehat{\{K_1,K_2\}} 
\eq
where $\widehat{K}$ is the contact Hamiltonian corresponding to the homogeneous Hamiltonian $K$ as in \eqref{contact1}. The bracket $\{\widehat{K}_1,\widehat{K}_2\}_J$ is equal to the {\it Jacobi bracket} of the contact Hamiltonians $\widehat{K}_1,\widehat{K}_2$; see e.g. \cite{libermann,bravetti,arnold} for the coordinate expressions of the Jacobi bracket. 
The Jacobi bracket is obviously bilinear and skew-symmetric. Furthermore, since the Poisson bracket satisfies the Jacobi-identity, so does the Jacobi bracket. However, the Jacobi bracket does {\it not} satisfy the Leibniz rule; i.e., in general the following equality does {\it not} hold
\bq
\{\widehat{K}_1,\widehat{K}_2 \cdot \widehat{K}_3\}_J = \{\widehat{K}_1,\widehat{K}_2\}_J \cdot \widehat{K}_3 + \widehat{K}_2 \cdot 
\{\widehat{K}_1,\widehat{K}_3\}_J
\eq
See also \cite{simoes} for additional information on the Jacobi bracket.

\subsection{Hamilton-Jacobi theory of Liouville and Legendre submanifolds}
Recall that any homogeneous Hamiltonian vector field $X_K$ on $\T^*\Q$ leaves invariant the Liouville form $\alpha$ and that Liouville submanifolds are maximal submanifolds on which $\alpha$ is zero. It follows that for any Liouville submanifold $\cL$ and any time $t \in \mR$ the evolution of $\cL$ along the homogeneous Hamiltonian vector field $X_K$ given as
\bq
\phi_t(\cL):= \{ \phi_t(z) \mid z \in \cL \} ,
\eq
where $\phi_t: \T^*\Q \to \T^*\Q$ is the flow map at time $t\geq 0$ of $X_K$, is also a Liouville submanifold. Applied to the Liouville submanifold characterizing the state properties of a thermodynamic system this means that the flow of a homogeneous Hamiltonian vector field transforms the Liouville submanifold to another Liouville submanifold at any time $t \geq 0$. For example, the Liouville submanifold corresponding to an ideal gas may be continuously transformed into the Liouville submanifold of a Van der Waals gas. This point of view was explored in \cite{mrugala1,Mrugala85, Mrugala93}.

Furthermore, cf. \eqref{22}, let $F(q_I,p_0,p_J) :=-p_0\widehat{F}(q_I, \frac{p_J}{-p_0})$, with $I \cup J=\{1,\cdots,n\}$, be the generating function of $\cL$, then it follows that for any $t\geq 0$ the generating function $G(q_I,p_0,p_J,t) :=-p_0\widehat{G}(q_I, \frac{p_J}{-p_0},t)$ of the transformed Liouville submanifold $\phi_t(\cL)$ satisfies the {\it Hamilton-Jacobi equation}
\bq
\begin{array}{l}
\frac{\partial G}{\partial t} + K(q_0,q_I, -\frac{\partial G}{\partial p_J}, p_0, \frac{\partial G}{\partial q_J}, p_J)=0 \\[3mm]
G(q_I,p_0,p_J,0) = F(q_I,p_0,p_J)
\end{array}
\eq
In case of the evolution of a general Lagrangian submanifold under the dynamics of a general Hamiltonian vector field this is classical  Hamilton-Jacobi theory (see e.g. \cite{abraham, arnold}), which directly specializes to Liouville submanifolds and homogeneous Hamiltonian vector fields. Furthermore, the generating functions $\widehat{G}(q_I, \gamma_J,t)$ of the corresponding Legendre submanifolds $\widehat{\phi_t(\cL)}$ satisfy the Hamilton-Jacobi equation (see also \cite{bravetti15})
\bq
\begin{array}{l}
\frac{\partial \widehat{G}}{\partial t} + \widehat{K}(q_0=\widehat{G} - \gamma_J\frac{\partial \widehat{G}}{\partial \gamma_J}, \,  q_J = - \frac{\partial \widehat{F}}{\partial \gamma_J}, \, \gamma_I = \frac{\partial \widehat{F}}{\partial q_I})=0 \\[3mm]
\widehat{G}(q_I,\gamma_J,0) = \widehat{F}(q_I,\gamma_J)
\end{array}
\eq
Note furthermore that $\widehat{\phi_t(\cL)} = \widehat{\phi_t} (\widehat{\cL})$, where $\widehat{\phi_t}$ is the flow map at time $t$ of the contact vector field $X_{\widehat{K}}$.
This implies as well the following result concerning {\it invariance} of Liouville and corresponding Legendre submanifolds, which will be one of the starting points for the definition of port-thermodynamic systems in the following section.

\begin{proposition}
\label{prop:invariant}\cite{mrugala2, libermann,valparaiso1}
Let $K: \T^*\Q \to \mR$ be homogeneous of degree $1$ in $p$, and let $\widehat{K}: \mP(T^*\Q) \to \mR$ be the corresponding contact Hamiltonian. Furthermore let $\cL \subset \T^*\Q$ be a Liouville submanifold, and $\widehat{\cL} \subset \mP(T^*\Q)$, with $\cL=\pi^{-1}(\widehat{\cL})$, the corresponding Legendre submanifold.
Then the following statements are equivalent:
\begin{enumerate}
\item
The homogeneous Hamiltonian vector field $X_{K}$ leaves $\mathcal{L}$ invariant.
\item
The contact vector field $X_{\widehat{K}}$ leaves $\widehat{\cL}$ invariant.
\item 
$K$ is zero on $\cL$.
\item
$\widehat{K}$ is zero on $\widehat{\cL}$.
\end{enumerate}
\end{proposition}

\section{Port-thermodynamic systems}
\label{sec:port-thermodynamic}
So far the geometric description of classical thermodynamics has been concerned with the {\it state properties}; starting from Gibbs' fundamental relation. Since these state properties are intrinsic to any thermodynamic system, they should be respected by any dynamics (thermodynamic processes). Hence any dynamics of an actual thermodynamic system should leave invariant the Liouville and Legendre submanifold characterizing the state properties \cite{mrugala2,Mrugala00,bravetti17,entropy}. Furthermore, desirably this should be the case for all possible state properties of the thermodynamic system, i.e., for all Liouville and Legendre submanifolds. This suggests that the dynamics on the canonical thermodynamic phase space $\mP(T^*Q)$ should be a {\it contact vector field} $X_{\widehat{K}}$, and the corresponding dynamics on $\T^*\Q$ should be a homogeneous Hamiltonian vector field $X_K$. 

Because of its simplicity, we first focus on the homogeneous Hamiltonian description. Consider a thermodynamic system with constitutive relations (state properties) specified by a Liouville submanifold $\cL \subset \T^*\Q$. Respecting the geometric structure means that the dynamics is a Hamiltonian vector field $X_K$ on $\T^*\Q$, with $K$ homogeneous of degree $1$ in the $p$-variables. Furthermore, since the state properties captured by $\cL$ are intrinsic to the system, the homogeneous Hamiltonian vector field $X_K$ should leave $\cL$ invariant. By Proposition \ref{prop:invariant} this means that the homogeneous Hamiltonian $K$ governing the dynamics should be zero on $\cL$. Furthermore, we will split $K$ into two parts, i.e.,
\bq
K^a + K^cu, \quad u \in \mR^m,
\eq
where $K^a: \T^*\Z \to \mR$ is the homogeneous Hamiltonian corresponding to the {\it autonomous} dynamics due to internal non-equilibrium conditions, while $K^c=(K^c_1, \cdots, K^c_m)$ is a row vector of homogeneous Hamiltonians (called {\it control} or {\it interaction} Hamiltonians) corresponding to dynamics arising from {\it interaction with the surrounding} of the system. This second part of the dynamics will be supposed to be affinely parametrized by a vector $u$ of {\it control}  or {\it input} variables (see however \cite{entropy} for an example of non-affine dependency). This means that all $(m+1)$ functions $K^a,K^c_1, \cdots, K^c_m$ are homogeneous of degree $1$ in $p$ and zero on $\cL$.

By invoking Euler's homogeneous function theorem (cf. Theorem \ref{Euler}) homogeneity of degree $1$ in $p$ means
\bq
\label{KaKc}
\begin{array}{rcll}
K^a &= & p_0\frac{\partial K^a}{\partial p_0} +  p_1\frac{\partial K^a}{\partial p_1} + \cdots + p_n\frac{\partial K^a}{\partial p_n}\\[3mm]
K^c & = & p_0\frac{\partial K^c}{\partial p_0} +  p_1\frac{\partial K^c}{\partial p_1} + \cdots + p_n\frac{\partial K^c}{\partial p_n} ,
\end{array}
\eq
where the functions $\frac{\partial K^a}{\partial p_i}$, as well as the elements of the $m$-dimensional row vectors of partial derivatives $\frac{\partial K^c}{\partial p_i}$, $i=0,1, \cdots, n$, are all homogeneous of degree $0$ in the $p$-variables. (Hence, as noted before, the dynamics of the extensive variables can be expressed as a function of the extensive variables and the intensive variables.)

The class of allowable autonomous Hamiltonians $K^a$ is further restricted by the {\it First and Second Law} of thermodynamics. Since the energy and entropy variables $E,S$ are among the extensive variables $q_0,q_1, \cdots,q_n$, let us denote $q_0=E, q_1=S$.
With this convention, the evolution of $E$ in the autonomous dynamics $X_{K^a}$ arising from non-equilibrium conditions is given by $\dot{E}=\frac{\partial K^a}{\partial p_0}$. Since by the First Law the energy of the system without interaction with the surrounding (i.e., for $u=0$) should be {\it conserved}, this implies that necessarily $\frac{\partial K^a}{\partial p_0}|_{\cL}=0$.
Similarly, $\dot{S}$ in the autonomous dynamics $X_{K^a}$ is given by $\frac{\partial K^a}{\partial p_1}$. Hence by the Second Law necessarily $\frac{\partial K^a}{\partial p_1}|_{\cL} \geq 0$. 

These two constraints need not hold for the control (interaction) Hamiltonians $K^c$. In fact, the analogous terms in the control Hamiltonians may be utilized to define natural {\it output} variables. First option is to define the output vector as the $m$-dimensional row vector ($p$ for power)
\bq
y_p=\frac{\partial K^c}{\partial p_0}
\eq
Then it follows that along the complete dynamics $X_K$ on $\cL$, with $K=K^a + K^cu$,
\bq
\frac{d}{dt} E= y_pu
\eq
Thus $y_p$ is the vector of {\it power-conjugate} outputs corresponding to the input vector $u$. We call the pair $(u, y_p)$ the {\it power port} of the system. 
Similarly, by defining the output vector as the $m$-dimensional row vector ($e$ for 'entropy flow')
\bq
y_{e}=\frac{\partial K^c}{\partial p_1}
\eq
it follows that along the dynamics $X_K$ on $\cL$
\bq
\frac{d}{dt} S \geq y_{e}u
\eq
Hence $y_{e}$ is the output vector which is conjugate to $u$ in terms of {\it entropy flow}. The pair $(u, y_{e})$ is called the {\it flow of entropy} port of the system. 

The above discussion is summarized in the following definition of a {\it port-thermodynamic system}.
\begin{definition}[\cite{entropy}]
\label{def:portthermo}
Consider the manifold of extensive variables $\Q$. A port-thermodynamic system on $\Q$ is a pair $(\cL,K)$, where $\cL \subset \T^*\Q$ is a Liouville submanifold describing the {\it state properties}, and $K=K^a + K^cu, u \in \mR^m,$ is a Hamiltonian on $\T^*\Q$, homogeneous of degree $1$ in $p$, and zero restricted to $\cL$, which generates the dynamics $X_K$.
Furthermore, let $q=(q_0,q_1, \cdots,q_n)$ with $q_0=E$ (energy), and $q_1=S$ (entropy). Then $K^a$ is required to satisfy $\frac{\partial K^a}{\partial p_0}|_{\cL}=0$ and $\frac{\partial K^a}{\partial p_1}|_{\cL}\geq0$. The {\it power conjugate output} vector of the port-thermodynamic system is defined as $y_p=\frac{\partial K^c}{\partial p_0}$, and the {\it entropy flow conjugate output} vector as $y_{e}=\frac{\partial K^c}{\partial p_1}$.
\end{definition}
Note that any port-thermodynamic system on $\T^*\Q$ immediately defines a corresponding system on the {\it thermodynamic phase space} $\mP(T^*\Q)$. Indeed, since $\cL \subset \T^*\Q$ is a Liouville submanifold it projects to a Legendre submanifold $\widehat{\cL} \subset \mP(T^*\Q)$. Furthermore, since $K$ is homogeneous of degree $1$ in $p$ it has the form $K(q,p)=-p_0\widehat{K}(q,\gamma)$, $\gamma_j=\frac{p_j}{-p_0}, j=1,\cdots,n$, with $\widehat{K}(q,\gamma)=\widehat{K}^a(q,\gamma)+ \widehat{K}^c(q,\gamma) u$ the contact Hamiltonian of the {\it energy representation}. This contact Hamiltonian is zero on $\widehat{\cL}$, and the dynamics $X_K$ projects to the contact vector field $X_{\widehat{K}}$ that leaves invariant $\widehat{\cL}$. Similarly, we can write $K(q,p)=-p_1\widehat{\widetilde{K}}(q,\tilde{\gamma})$, $\tilde{\gamma}_j=\frac{p_j}{-p_1}, j=0,2\cdots,n$, with $\widehat{\widetilde{K}}(q,\tilde{\gamma})$ the contact Hamiltonian of the {\it entropy representation}.
Furthermore, by Euler's theorem both the power conjugate output $y_p$ and the entropy flow conjugate output $y_{e}$ are homogeneous of degree $0$, and thus project to functions on $\mP(T^*\Q)$. Finally, in the energy representation we can rewrite the power conjugate output as 
\bq
y_p=\frac{\partial K^c}{\partial p_0}= \sum_{\ell=1}^n \gamma_\ell \frac{\partial \widehat{K}^c}{\partial \gamma_\ell}(q,\gamma)-\widehat{K}^c(q,\gamma)
\eq
Similarly for the entropy flow conjugate output $y_{e}=\frac{\partial K^c}{\partial p_1}=\sum_{\ell=0,2}^n \tilde{\gamma}_\ell \frac{\partial \widehat{\widetilde{K}}^c}{\partial \tilde{\gamma_\ell}}(q,\tilde{\gamma})-\widehat{\widetilde{K}}^c(q,\tilde{\gamma})$.
Finally note that the constraints imposed on $K^a$ by the First and Second law can be written in contact-geometric terms as
\bq
\label{constraints-contact}
\begin{array}{l}
\left(\sum_{\ell=1}^n \gamma_\ell \frac{\partial \widehat{K}^a}{\partial \gamma_\ell}(q,\gamma)-\widehat{K}^a(q,\gamma)\right) |_{\widehat{\cL}} =0 \\[2mm]
\left(\sum_{\ell=0,2}^n \tilde{\gamma}_\ell \frac{\partial \widehat{\widetilde{K}}^a}{\partial \tilde{\gamma_\ell}}(q,\gamma)-\widehat{\widetilde{K}}^a(q,\tilde{\gamma}) \right)|_{\widehat{\cL}} \geq 0
\end{array}
\eq

\begin{example}[Gas-piston-damper system]
\label{ex:gaspiston}
Consider a gas in a thermally isolated compartment closed by a piston. Assume the thermodynamic properties of the system to be fully covered by the properties of the gas. The extensive variables are given by energy $E$, entropy $S$, volume $V$, and momentum of the piston $\pi$. The state properties of the system are described by the Liouville submanifold $\cL$ with generating function (in energy representation)
$-p_E\left(U(S,V) + \frac{\pi^2}{2m} \right)$,
where $U(S,V)$ is the energy of the gas, and $\frac{\pi^2}{2m}$ the kinetic energy of the piston with mass $m$. This defines the state properties
\bq
\begin{array}{rcl}
\cL & = & \{(E,S,V,\pi,p_E,p_S,p_V,p_{\pi}) \mid E=U(S,V) + \frac{\pi^2}{2m}, 
\\[2mm]
&& p_S=-p_E\frac{\partial U}{\partial S}(S,V), p_V=-p_E\frac{\partial U}{\partial V}(S,V), p_{\pi}= -p_E \frac{\pi}{m} \}
\end{array}
\eq
Assume the damper is linear with damping constant $d$. The dynamics of the gas-piston-damper system, with piston actuated by a force $u$, is given by $X_K$, where the homogeneous Hamiltonian $K:\T^*\mR^4 \to \mR$ is given as
\bq
K= p_V\frac{\pi}{m} +p_{\pi}\left(-\frac{\partial U}{\partial V} -d\frac{\pi}{m}\right) +p_S \frac{d (\frac{\pi}{m})^2}{\frac{\partial U}{\partial S}} + \left(p_{\pi} + p_E \frac{\pi}{m}\right)u,
\eq
which is zero on $\cL$. The power-conjugate output $y_p=\frac{\pi}{m}$ is the velocity of the piston. 
In energy representation the description projects to the thermodynamic phase space $\mP(T^*\mR^4) = \{(E,S,V,\pi,T,-P,v)\}$, with $\gamma_S=T$ (temperature), $\gamma_V=-P$ (pressure), and $\gamma_{\pi}=v$ (velocity of the piston) as follows. First note that $\cL$ projects to the Legendre submanifold
\bq
\widehat{\cL} = \{(E,S,V,\pi,T,-P,v) \mid E=U(S,V) + \frac{\pi^2}{2m},T=\frac{\partial U}{\partial S}, -P=\frac{\partial U}{\partial V}, v= \frac{\pi}{m} \}
\eq
Furthermore, $K=-p_E\widehat{K}$ with
\bq
\widehat{K}= -P\frac{\pi}{m} + v \left(- \frac{\partial U}{\partial V} - d \frac{\pi}{m} \right) + T \frac{d (\frac{\pi}{m})^2}{\frac{\partial U}{\partial S}} + (v- \frac{\pi}{m})u
\eq
This yields the following dynamics of the extensive variables
\bq
\begin{array}{rcl}
\dot{E} & = & \frac{\pi}{m}u \\
\dot{S} & = & d (\frac{\pi}{m})^2 / \frac{\partial U}{\partial S} \quad (\geq 0)\\
\dot{V} & = & \frac{\pi}{m} \\
\dot{\pi} & = & - \frac{\partial U}{\partial V} - d \frac{\pi}{m} +u ,
\end{array}
\eq
while the intensive variables satisfy $\dot{T}= -\frac{\partial \widehat{K}}{\partial S}, - \dot{P}= -\frac{\partial \widehat{K}}{\partial V}, \dot{v}=-\frac{\partial \widehat{K}}{\partial \pi}$. Similarly for the entropy representation.
\end{example}
In {\it composite} thermodynamic systems, there is typically no {\it single} energy or entropy. In this case the {\it sum} of the energies needs to be conserved by the autonomous dynamics, and likewise the {\it sum} of the entropies needs to be increasing. A simple example is the following; see \cite{entropy} for further information.

\begin{example}[Heat exchanger]\label{ex:heatexch}
Consider two heat compartments, exchanging a heat flow through a conducting wall according to Fourier's law. Each heat compartment is described by an entropy $S_i$ and energy $E_i$, $i=1,2$, corresponding to the Liouville submanifolds
\bq
\cL_i = \{(E_i,S_i,p_{E_i}, p_{S_i}  \mid  E_i=E_i(S_i), p_{S_i} = - p_{E_i}E'_i(S_i)\}, \quad E'_i(S_i)\geq 0
\eq
Taking $u_i$ as the incoming heat flow into the $i$-th compartment corresponds to
\bq
K^c_i= p_{S_i}\frac{1}{E'_i(S_i)} + p_{E_i},
\eq 
while $K^a_i=0$. This defines the flow of entropy conjugate outputs as $y_{ei} = \frac{1}{E'_i(S_i)}$ (reciprocal temperatures). The conducting wall is described by the interconnection equations (with $\lambda$ Fourier's conduction coefficient)
\bq
-u_1= u_2 = \lambda(\frac{1}{y_{e1}} - \frac{1}{y_{e2}}),
\eq
relating the incoming heat flows $u_i$ and reciprocal temperatures $y_i$, $i=1,2$, at both sides of the conducting wall.
This leads to (setting $E(S_1,S_2):=E_1(S_1) + E_2(S_2), p_{E_1}=p_{E_2}=:p_E$, cf. \cite{entropy}) to the autonomous dynamics generated by the homogeneous Hamiltonian
\bq
K^a:= K^c_1u_1 + K^c_2u_2 = \lambda \left(p_{S_1}\frac{1}{E'(S_1)} + p_{S_2}\frac{1}{E'(S_2)}\right) \left(E'(S_2) - E'(S_1)\right)
\eq
Hence the total entropy on the Liouville submanifold 
\bq
\cL \! = \! \{(E,S_1,S_2, p_E, p_{S_1},p_{S_2})| E=E_1 +E_2, p_{S_1}=-p_EE_1'(S_1), p_{S_2}=-p_EE_2'(S_2) \}
\eq
satisfies
\bq
\frac{d}{dt}({S}_1 + {S}_2) = \lambda (\frac{1}{E_1'(S_1)} - \frac{1}{E_2'(S_2)})(E_2'(S_2) - E_1'(S_1)) \geq 0
\eq
\end{example}

Interestingly, while the Hamiltonians in standard Hamiltonian systems (such as in mechanics) represent {\it energy}, the Hamiltonians $K$ in the above examples are {\it dimensionless} (in the sense of dimensional analysis). This holds in general. Furthermore, it can be verified that the {\it contact Hamiltonian} of its projected dynamics (a contact vector field) has dimension of {\it power} in case of the energy representation (with intensive variables $T,-P$), and has dimension of {\it entropy flow} in case of the entropy representation (with intensive variables $\frac{1}{T}, \frac{P}{T}$). Together with the fact that the dynamics of a thermodynamic system is captured by the dynamics {\it restricted} to the invariant Liouville submanifold, this emphasizes that the interpretation of the Hamiltonian dynamics $X_K$ is rather {\it different} from the Hamiltonian formulation of mechanical (or other physical) systems. 

Finally, let us recall the well-known correspondence \cite{libermann,arnold} between {\it Poisson brackets} of Hamiltonians $K_1,K_2$, and {\it Lie brackets} of their corresponding Hamiltonian vector fields, i.e.,
\bq
[X_{K_1},X_{K_2}] =X_{\{K_1,K_2\}}
\eq
In particular, this property implies that if the homogeneous Hamiltonians $K_1,K_2$ are zero on the Liouville submanifold $\mathcal{L}$, and thus by Proposition \ref{prop:invariant} the homogeneous Hamiltonian vector fields $X_{K_1},X_{K_2}$ are tangent to $\mathcal{L}$, then also $[X_{K_1},X_{K_2}]$ is tangent to $\mathcal{L}$, and therefore the Poisson bracket $\{K_1,K_2\}$ is also zero on $\mathcal{L}$. Together with Proposition \ref{app:poisson} this was crucially used in the controllability and observability analysis of port-thermodynamic systems in \cite{toulouse}.

\section{Homogeneity in the extensive variables and Gibbs-Duhem relation}
\label{sec:homex}
In many thermodynamic systems, when taking into account {\it all} extensive variables, there is an additional form of homogeneity; now with respect to the {\it extensive variables} $q$. To start with, consider a Liouville submanifold $\cL$ with generating function $-p_0\widehat{F}(q_1, \cdots, q_n)$. Recall that if $q_0$ denotes the energy variable, then $\widehat{F}(q_1, \cdots,q_n)$ equals the energy $q_0$ expressed as a function of the other extensive variables $q_1,\cdots,q_n$. Assume that the manifold of extensive variables $\Q$ is the {\it linear space}\footnote{Homogeneity can be generalized to manifolds using the theory developed in \cite{libermann}.} $\Q = \mR^{n+1}$. Homogeneity with respect to the extensive variables means that the function $\widehat{F}$ is homogeneous of degree $1$ in $q_1, \cdots,q_n$. This implies by Euler's theorem (Theorem \ref{Euler}) that $\widehat{F}= \sum_{j=1}^n q_j\frac{\partial \widehat{F}}{\partial q_j}$. Hence on the corresponding Legendre submanifold $\widehat{\cL}=\pi(\cL)$ we have $\widehat{F}= \sum_{j=1}^n  \gamma_j q_j$, and thus 
\bq
d\widehat{F} = \sum_{j=1}^n \gamma_j dq_j+ \sum_{j=1}^n q_j d\gamma_j 
\eq
By Gibbs' relation this implies that on $\widehat{\cL}$
\bq
\label{gd}
\sum_{j=1}^n q_j d\gamma_j =0,
\eq
which is known as the {\it Gibbs-Duhem relation}; see e.g. \cite{kondepudi, gromov3}. The relation implies that the intensive variables $\gamma_j$ on $\widehat{\cL}$ are {\it dependent}.

More generally this can be formulated in the following {\it geometric} way.
\begin{definition}
Let $\Q = \mR^{n+1}$ with linear coordinates $q$. A Liouville submanifold $\cL \subset \T^*\mR^{n+1}$ is {\it homogeneous with respect to the extensive variables} $q$ if
\bq
(q_0,q_1, \cdots,q_n, p_0, \cdots,p_n) \in \cL \Rightarrow (\mu q_0, \mu q_1, \cdots,\mu q_n, p_0, \cdots,p_n) \in \cL
\eq
for all $0 \neq \mu \in \mR$. 
\end{definition}
Using the same theory as exploited before for homogeneity with respect to the $p$-variables, cf. Proposition \ref{prop:liouville}, homogeneity of $\cL$ with respect to $q$ is equivalent to the vector field $W:=\sum_{i=0}^n q_i \frac{\partial}{\partial q_i}$ being {\it tangent} to $\cL$. Hence, using the same argumentation as in Proposition \ref{prop:liouville}, not only the Liouville form $\alpha = \sum_{i=0}^n p_i dq_i$ is zero on $\cL$, but also the one-form
\bq
\label{gdg}
\beta:=  \sum_{i=0}^n q_i dp_i
\eq
This could be called the {\it generalized} Gibbs-Duhem relation.
\begin{proposition}
The Liouville submanifold $\cL$ is homogeneous with respect to the extensive variables $q$ if and only if $\beta=  \sum_{i=0}^n q_i dp_i$ is zero on $\cL$. Let $\cL$ have generating function $-p_0\widehat{F}(q_I,\gamma_J)$ for some partitioning $\{1,\cdots,n\}=I \cup J$. Then $\cL$ is homogeneous with respect to the extensive variables $q$ if and only if if $I$ is non-empty and $\widehat{F}(q_I,\gamma_J)$ is homogeneous of degree $1$ in $q_I$. Furthermore, if $\cL$ is homogeneous with respect to the extensive variables $q$, then
\bq
\label{sumi}
\sum_{i=0}^n q_i p_i =0, \quad \mbox{ for all } (q,p) \in \cL
\eq
\end{proposition}
\begin{proof}
As mentioned above, the first statement follows from the same reasoning as in Proposition \ref{prop:liouville}, swapping the $p$ and $q$ variables. Equivalence of homogeneity of $\cL$ with respect to $q$ to $\widehat{F}(q_I,\gamma_J)$ being homogeneous of degree $1$ in $q_I$ directly follows from the expression of $\cL$ in \eqref{20} in case $I \neq \emptyset$, while clearly homogeneity of $\cL$ fails if $I = \emptyset$. Finally, if both $\alpha=\sum_{i=0}^n p_i dq_i$ and $\beta=\sum_{i=0}^n q_i dp_i$ are zero on $\cL$, then $d (\sum_{i=0}^n q_i p_i)$ is zero on $\cL$. Hence $\sum_{i=0}^n q_i p_i$ is constant on $\cL$. Since $Z=\sum_{i=0}^n p_i \frac{\partial}{\partial p_i}$ is tangent to $\cL$ necessarily this constant is zero.
\end{proof}
\begin{remark} In a contact-geometric setting, an identity similar to \eqref{sumi} was noticed in \cite{hoang}. A related scenario, explored in \cite{brayton}, is the case that $\cL$ is a Lagrangian submanifold which is {\it non-mixing}: there exists a partitioning $\{0,1,\cdots n\}$ $= I \cup J$ such that $q_J=q_J(q_I), \, p_I=p_I(p_J)$ for all $(q_I,q_J,p_I,p_J) \in \cL$.
Then $\cL$ being Lagrangian amounts to
\bq
\frac{\partial q_J}{\partial q_I} = - \left(\frac{\partial p_I}{\partial p_J}\right)^\top
\eq
Since the left-hand side only depends on $q_I$ and the right-hand side only on $p_J$, this means that both sides are {\it constant}, implying that $q_J=Aq_I, p_I= - A^\top p_J$ for some matrix $A$. Hence $\cL$ is obviously satisfying \eqref{sumi}, and is actually the product of two orthogonal {\it linear} subspaces; one in $\Q=\mR^{n+1}$ and the other in the dual space $\Q^*=\mR^{n+1}$.
\end{remark}
Homogeneity of $\cL$ with respect to $q$ has the following classical implication. Consider again the case of a generating function $F(q,p)=-p_0\widehat{F}(q_1, \cdots, q_n)$ for $\cL$, with $q_0$ being the energy variable. Since $\widehat{F}$ is homogeneous of degree $1$ we may define for $q_1 \neq 0$
\bq
\label{eps}
\bar{F}(\epsilon_2, \cdots,\epsilon_n):= \widehat{F}(1,\frac{q_2}{q_1}, \cdots, \frac{q_n}{q_1})= \frac{1}{q_1}\widehat{F}(q_1,\cdots,q_n), \; \epsilon_j:=\frac{q_j}{q_1}, \; j=0,2, \cdots,n
\eq
Equivalently, $\widehat{F}(q_1,\cdots,q_n)= q_1 \bar{F}(\epsilon_2, \cdots,\epsilon_n)$, where the function $\bar{F}$ is known as the {\it specific energy} \cite{kondepudi}. 

Geometrically this means the following. By homogeneity with respect to the $p$-variables the Liouville submanifold $\cL \subset \T^*\mR^{n+1} $ is projected to the Legendre submanifold $\widehat{\cL} \subset \mR^{n+1} \times \mP(\mR^{n+1})$, where $\mP(\mR^{n+1})$ is the $n$-dimensional projective space. Subsequently, by homogeneity with respect to the $q$-variables $\widehat{\cL} \subset \mR^{n+1} \times \mP(\mR^{n+1})$ is projected to a submanifold $\bar{\cL} \subset \mP(\mR^{n+1})\times \mP(\mR^{n+1})$. In coordinates the expression of $\bar{\cL}$ is given as follows. Start from the expression of $\widehat{\cL}$ as given in \eqref{20}. Using the identities
\bq
\begin{array}{l}
q_0=q_1 \bar{F}(\epsilon_2, \cdots, \epsilon_n) \Leftrightarrow \epsilon_0 = \bar{F}(\epsilon_2, \cdots, \epsilon_n) \\[2mm]
\gamma_1=\frac{\partial \widehat{F}}{\partial q_1} = \bar{F}(\epsilon_2, \cdots, \epsilon_n) - q_1 \sum_{\ell =2}^n \frac{\partial \bar{F}}{\partial \epsilon_{\ell}} \frac{q_{\ell}}{q_1^2} = \bar{F}(\epsilon_2, \cdots, \epsilon_n) - \sum_{\ell =2}^n \epsilon_{\ell}\frac{\partial \bar{F}}{\partial \epsilon_{\ell}} \\[2mm]
\gamma_j = \frac{\partial \widehat{F}}{\partial q_j} = \frac{\partial (q_1\bar{F})}{\partial q_j}= \frac{\partial \bar{F}}{\partial \epsilon_j}, \quad j=2, \cdots,n
\end{array}
\eq
the description \eqref{20} amounts to
\bq
\begin{array}{l}
\bar{\cL}= \{(\epsilon_0,\epsilon_2, \cdots, \epsilon_n, \gamma_1,\cdots, \gamma_n) \mid \epsilon_0 = \bar{F}(\epsilon_2, \cdots, \epsilon_n), \\[2mm]
 \gamma_1=\bar{F}(\epsilon_2, \cdots, \epsilon_n) -  \sum_{\ell =2}^n \epsilon_\ell \frac{\partial \bar{F}}{\partial \epsilon_\ell}, \, \gamma_2= \frac{\partial \bar{F}}{\partial \epsilon_2}, \cdots, \gamma_n=\frac{\partial \bar{F}}{\partial \epsilon_n} \},
 \end{array}
 \eq
where 
\bq
F(q,p)= -p_0\widehat{F}(q)= -p_0q_1 \bar{F}(\epsilon_2, \cdots, \epsilon_n), \quad \epsilon_j:=\frac{q_j}{q_1}, \; j=0,2, \cdots,n
\eq
Similar expressions hold in the general case that the generating function for $\widehat{\cL}$ is given by $\widehat{F}(q_I,\gamma_J)$ for some partitioning $\{1, \cdots,n \} =I \cup J$.
%
\medskip

Furthermore, if the state properties captured by $\cL$ are homogeneous with respect to $q$, it is natural to require the dynamics to be homogeneous with respect to $q$ as well. Thus one requires the Hamiltonian $K(q,p)$ governing the dynamics to be homogeneous of degree $1$, not only with respect to $p$, but also with respect to $q$, i.e.,
\bq
K(\mu q,p)= \mu K(q,p), \quad \mbox{ for all } 0 \neq \mu \in \mR
\eq
Equivalently (analogously to Proposition \ref{prop:homham}) one requires $X_K$ to satisfy 
\bq
\mL_{X_K}\beta=0
\eq
Similarly to Proposition \ref{prop:projection}, this implies
\bq
[X_K,W]=0, \quad W=\sum_{i=0}^n q_i \frac{\partial}{\partial q_i}
\eq
Hence the flow of $X_K$ commutes both with the flow of the Euler vector field $Z=\sum_{i=0}^n p_i \frac{\partial}{\partial p_i}$ and with the vector field $W=\sum_{i=0}^n q_i \frac{\partial}{\partial q_i}$. 

We have seen before that projection along $Z$ yields the {\it contact vector field} $X_{\widehat{K}}$, with $K(q,p)= -p_0\widehat{K}(q,\gamma), \, \gamma_j= \frac{p_j}{-p_0}, j=1, \cdots,n$, where $(q,\gamma) \in \mR^{n+1} \times \mP(\mR^{n+1})$. 
Subsequent projection along $W$ to the reduced space $\mP(\mR^{n+1})\times \mP(\mR^{n+1})$ can be computed as follows. First write as above
\bq
\widehat{K}(q,\gamma)=q_1 \bar{K}(\epsilon,\gamma), \quad \epsilon_j= \frac{q_j}{q_1}, \quad j=0, 2, \cdots,n
\eq
Then compute, analogously to \eqref{identities},
\bq
\begin{array}{ll}
\frac{\partial \widehat{K}}{\partial q_1} = \bar{K} - \sum_{\ell=0,2}^n \epsilon_\ell \frac{\partial \bar{K}}{\partial \epsilon_\ell} & \\[2mm]
\frac{\partial \widehat{K}}{\partial q_j} = \frac{\partial \bar{K}}{\partial \epsilon_j}, \quad & j=0,2 \cdots,n \\[2mm]
\frac{\partial \widehat{K}}{\partial \gamma_j}=q_1 \frac{\partial \bar{K}}{\partial \gamma_j}, \quad & j=1, \cdots,n 
\end{array}
\eq
Combining, analogously to \eqref{dot}, with the expression
\bq
\dot{\epsilon}_j = \frac{\dot{q}_j}{q_1} - \frac{q_j}{q_1^2}\dot{q}_1,
\eq
this yields the following $2n$-dimensional dynamics on the {\it reduced thermodynamic phase space} $\mP(\mR^{n+1}) \times \mP(\mR^{n+1})$
\bq
\label{projected}
\begin{array}{rcll}
\dot{\epsilon}_j & = & \frac{\partial \bar{K}}{\partial \gamma_j} - \epsilon_j \left(\sum_{\ell=1}^n \gamma_\ell \frac{\partial \bar{K}}{\partial \gamma_\ell} - \bar{K} \right) , \quad &  j=0,2, \cdots, n\\[3mm]
\dot{\gamma}_j & = & - \frac{\partial \bar{K}}{\partial \epsilon_j} + \gamma_j \left(\sum_{\ell=0,2}^n \epsilon_\ell \frac{\partial \bar{K}}{\partial \epsilon_\ell} - \bar{K} \right), & j=1,2 \cdots, n,
\end{array}
\eq
where $\bar{K}$ is determined by 
\bq
K(q,p) = - p_0q_1 \bar{K}(\epsilon,\gamma), \quad \epsilon=\left(\frac{q_0}{q_1},\frac{q_2}{q_1} \cdots, \frac{q_n}{q_1}\right), \; \gamma= \left(\frac{p_1}{-p_0}, \cdots, \frac{p_n}{-p_0}\right)
\eq
Obviously, if $q_0$ represents {\it entropy} the same expressions hold with different interpretation of $\epsilon_0, \epsilon_2, \cdots, \epsilon_n$.

Note that the $2n$-dimensional dynamics \eqref{projected} consists of standard Hamiltonian equations with respect to the Hamiltonian $\bar{K}$, together with extra terms. In view of \eqref{constraints-contact}, the first part of these extra terms for the autonomous term $\bar{K}^a$, i.e., $\sum_{\ell=1}^n \gamma_\ell \frac{\partial \bar{K}^a}{\partial \gamma_\ell} - \bar{K}^a$, is zero on $\cL$.

As a final remark it can be noted that while the above reduction from $\cL$ and $X_K$ to $\bar{\cL}$ and the dynamics \eqref{projected} was done via $\widehat{\cL}$ and $X_{\widehat{K}}$ (the contact-geometric description on the thermodynamic phase space), the same outcome is obtained by instead {\it first} projecting onto $\mP(\mR^{n+1}) \times \mR^{n+1}$ along $W$, and {\it then} projecting onto $\mP(\mR^{n+1}) \times \mP(\mR^{n+1})$ along $Z$. Said otherwise, this alternative route involves a {\it different} intermediate contact geometric description on the contact manifold $\mP(\mR^{n+1}) \times \mR^{n+1}$ with coordinates $\epsilon_0,\epsilon_2, \cdots, \epsilon_n,p_0,\cdots,p_n$.  

\section{Conclusions}
The geometric formulation of classical thermodynamics gives rise to a specific branch of symplectic geometry, coined as {\it Liouville geometry}, which is closely related to contact geometry. A detailed treatment of Liouville submanifolds and their generating functions has been provided. The same has been done for homogeneous Hamiltonian vector fields, extending the treatment in e.g. \cite{arnold, arnold-contact, libermann}. For the formulation of the Weinhold and Ruppeiner metrics in this setting we refer to \cite{entropy}. 
The interpretation of the resulting Hamiltonian formulation of port-thermodynamic systems turns out to be rather different from Hamiltonian formulations of other parts of physics, such as mechanics. In particular, the state properties of the thermodynamic system define a Liouville submanifold, which is left invariant by the Hamiltonian dynamics. Furthermore, the Hamiltonian is dimensionless, while its corresponding contact Hamiltonians have dimension of power (energy representation) or entropy flow (entropy representation). An open modeling problem concerns the determination of the Hamiltonian governing the dynamics. A partial answer is given in \cite{entropy}, where it is shown how the Hamiltonian of a thermodynamic system can be derived from the Hamiltonians of the constituent thermodynamic subsystems. In Section \ref{sec:homex} another type of homogeneity has been considered; this time with respect to the extensive variables, corresponding to the classical Gibbs-Duhem relation. It has been shown how this gives rise to a further projected dynamics on the product of the $n$-dimensional projective space with itself. The precise geometric interpretation and properties of the reduced dynamics \eqref{projected} deserve further study.

\section*{Acknowledgements}
\noindent
I thank Bernhard Maschke, Universit\'e de Lyon-1, France, for ongoing collaborations that stimulated the writing of the present paper.

\end{document}